\documentclass[letterpaper, conference]{IEEEtran}
\IEEEoverridecommandlockouts

\usepackage{cite}
 \usepackage{amsmath,amssymb,amsfonts}
 \usepackage{algorithmicx}
 \usepackage{algpseudocode}
 \usepackage{algorithm}
 \usepackage{amsthm}
 \usepackage{graphicx}
 \usepackage{enumitem}
 \usepackage{xcolor}
 \usepackage{float}

 \newtheorem{prop}{Proposition}
\newtheorem{definition}{Definition}
 \newtheorem{remark}{Remark}
 \newtheorem{case}{Case}

\usepackage{subfigure}

\title{Continuous-time Signal Temporal Logic Planning with Control Barrier Function}

\author{Guang Yang, Roberto Tron and Calin Belta}

\begin{document}
\maketitle

\begin{abstract}
    Temporal Logic (TL) guided control problems have gained interests in recent years. By using the TL, one can specify a wide range of temporal constraints on the system and is widely used in cyber-physical systems. On the other hand, Control Barrier Functions have also gained interests in the context of safety critical applications. However, most of the existing approaches only focus on discrete-time dynamical systems.
    In this paper, we propose an offline trajectory planner for linear systems subject to safety and temporal specifications. Such specifications can be expressed as logical junctions or disjunctions of linear CBFs, or as STL specifications with linear predicates. Our planner produces trajectories that are valid in continuous time,  while assuming only discrete-time control updates and arbitrary time interval in the STL formula.  
    
    Our planner is based on a Mixed Integer Quadratic Programming (MIQP) formulation, where the linear STL predicates are encoded as set of linear constraints to guarantee satisfaction at on a finite discrete set of time instants, while we use CBFs to derive constraints that guarantee continuous satisfaction between time instants. Moreover, we have shown the predicates can be encoded as time-based CBF constraints for system with any relative degrees. We validate our theoretical results and formulation through numerical simulations.
\end{abstract}

\section{Introduction}
Temporal Logic Based control has been widely used in the context of persistent surveillance \cite{leahy2016persistent}, traffic control \cite{sadraddini2016model} and distributed sensing \cite{serlin2018distributed}. Originated from the context of formal methods in model checking\cite{baier2008principles}, there are a wide range of Temporal Logics (TL), such as Linear Temporal Logic (LTL) and Computation Tree Logic (CTL) are used to describe a rich set of specifications of system behaviors. For applications that require to define real values with bounded time constraints, the Signal Temporal Logic (STL) \cite{maler2004monitoring} was introduced as a means of achieving such desired system behavior. 

The notion of space robustness of STL for real-valued signal was first introduced in \cite{donze2010robust}. The proposed method of evaluating the space robustness on a continuous signal occurs in discrete time. Later, the authors in \cite{raman2014model}, \cite{raman2015reactive} and \cite{sadraddini2015robust} use mixed-integer encoding method to encode STL robustness in discrete time steps. This type of approach has a major drawback that there is no guarantee of this satisfaction of the formula in between two sampled time steps. In \ref{sec:example2}, we compare the discrete-time robustness based STL planning with our proposed method for safety. In Fig \ref{fig:f2}, it clearly shows the continuous-time trajectory violates safety constraints due to the sampling limitation. In real systems, there are usually limitation on the controller update rate. Therefore, the time constraint for each STL predicate is limited to the controller update instants. In \ref{sec:example3}, we illustrate an example where the time constraint from a STL predicate causes asynchronous controller updates between the desired update instants and actual update instants. More details will be explained in later section.

The Control Barrier Functions (CBF) was introduced in \cite{wieland2007constructive}. The main goal is to ensure the system trajectory stays \textit{forward invariant} under some pre-defined safety sets. It is later extended to the applications of Adaptive Cruise Control \cite{ames2014control}, swarm manipulation \cite{borrmann2015control}, heterogeneous multi-Robot manipulation\cite{wang2016safety} and bipedal robotic walking\cite{hsu2015control}, where the control problem is formulated as a Quadratic Program (QP) with the CBF constraints. The formulated optimization problem is solved point-wise in time. This type of approach also suffer from the drawback that is mentioned earlier due to the controller update rate, i.e., there is no guarantee the CBF constraints will hold true in between the two controller updates. There are some works, such as \cite{ghaffari2018safety}, address the issue, but the approach limits to controller design rather than the full trajectory planning. The CBF type of approach has also extended to machine learning based control \cite{chen2017enhancing},\cite{wang2018safe} with partially known system dynamics and unknown disturbances. 

In this paper, we are interested in linear systems with continuous dynamics and linear STL predicates. We would like to synthesize a sequence of control inputs subject to the STL specification by formulating an optimization based planning problem. The control inputs are applied discretely in a zeroing order hold (ZOH) manner, i.e., the control holds at as constant value in-between two updates. There are three major contributions from this paper. First, we propose a novel integer encoding method for obtaining the lower bound of a linear CBF constraint for a fixed time interval. Second, we overcome the drawbacks from discrete-time STL robustness based control by directly encoding the STL robustness as time-based CBF constraints to achieve continuous-time satisfactory. Third, we illustrate a way to resolve the issue of asynchronous controller update time due to the limitation of sampling rate on the actual system. 

\section{Preliminaries}
\subsection{Notation}
We use $\mathbb{Z}$
and $\mathbb{R}^n$ to denote the set of integers and the set of real numbers in $n$ dimensions, respectively. We define $x[t]$ as the real value of $x$ at time instance $t$. The Lie derivative of a smooth function $h(x(t))$ along dynamics $\dot{x}(t)=f(x(t))$ is denoted as $\pounds_{f} h(x) := \frac{\partial h(x(t))}{\partial x(t)} f(x(t))$. We use $\pounds_{f}^{r_b} h(x)$ to define a Lie derivative of higher order $r_b$, where $r_b \geq 0$. A function $f: \mathbb{R}^n \mapsto \mathbb{R}^m$ is called \textit{Lipschitz continuous} on $\mathbb{R}^n$ if there exists a positive real constant $L \in \mathbb{R}^+$, such that $\|f(y)-f(x)\| \leq L \|y-x\|, \forall x,y \in \mathbb{R}^n$. Given a continuously differentiable function $h:\mathbb{R}^n \mapsto \mathbb{R}$, we denote $h^{r_b}$ as its $r_b$-th derivative with respect to time $t$. A continuous function $\alpha:[-b,a) \mapsto [-\infty,\infty)$, for some $a>0, b>0$, belong to extended class $K$ if $\alpha$ is strictly increasing on $\mathbb{R}^+$ and $\alpha(0)=0$. 

\subsection{Dynamical System and Safety Set}
Consider an affine control system:
\begin{equation}
\begin{aligned}
    \dot{x} &= Ax+Bu, \\ \label{continousSystem}
    y &= Cx,
\end{aligned}
\end{equation}
where $x\in \mathbb{R}^n$, $A \in \mathbb{R}^{n\times n}$, $B\in \mathbb{R}^{n\times m}$, $C\in \mathbb{R}^{1\times n}$ and $u\in \mathbb{R}^{m}$.
 
Given a continuously differentiable function $h:\mathbb{R}^n \mapsto \mathbb{R}$ and dynamics \eqref{continousSystem}, with $f(x)=Ax, g(x)=B$, the relative degree $r_b \geq 1$ is defined as the smallest natural number such that $\pounds_{g} \pounds_{f}^{r_b-1} h(x) u \neq 0$. The time derivative of $h$ are related to the Lie derivatives by:

\begin{equation} \label{eq:rbh}
h^{r_b}(x) = \pounds_{f}^{r_b}h(x) + \pounds_{g}\pounds_{f}^{r_b-1}h(x)u.
\end{equation}
Next, we define a closed safety set $C$:
\begin{equation}\label{eq:safetySet}
\begin{aligned}
C &= \{x \in{\mathbb{R}^n}|h^{r_b}(x)  \geq 0 \}. \\
\partial{C} &= \{x \in \mathbb{R}^n|h^{r_b}(x)  = 0 \}, \\
Int(C) &= \{x \in \mathbb{R}^n|h^{r_b}(x)  > 0 \},
\end{aligned}
\end{equation}
where $\partial{C}$ is the boundary of the set and $Int(C)$ is the interior of the set. 

\subsection{Exponential Control Barrier Function}
To ensure forward invariance for systems with higher relative degrees, the author in \cite{nguyen2016exponential} introduces the notion of Exponential Control Barrier Function (ECBF). Before formally reviewing its definition, a transverse variable is defined as
\begin{align}
\xi_b(x) = \left[ \begin{matrix} h(x), \dot{h}(x), ...,h^{r_b}(x) \end{matrix} \right]^T.
\end{align}
with virtual control input
\begin{equation}
\mu = (\pounds_{g}\pounds_{f}^{r_b-1}h(x))^{-1}(\mu-\pounds_{f}^{r_b}h(x)).
\end{equation}
The input-output linearized system is defined as
\begin{align*}
&\dot{\xi}_b(x) = A_b \xi_b(x) + B_b \mu, \\
&y = C_b \xi_b(x)= h(x),
\end{align*}

\begin{definition}
(\textit{Exponential Control Barrier Function}) 
Consider the dynamical system \eqref{continousSystem}, the safety set $C$ defined in \eqref{eq:safetySet} and $h(x)$ with relative degree $r_b \geq 1$. Then $h(x)$ is an exponential control barrier function (ECBF) if there exists $K_b \in \mathbb{R}^{1\times r_b}$, such that
\begin{equation} \label{eq:ecbfConstraint}
\inf \limits_{u\in U} [\pounds_{f}^{r_b} h(x)+\pounds_{g}\pounds_{f}^{r_b-1} h(x)u+ K_b \xi_b(x) ] \geq 0, \forall x \in Int(C).
\end{equation}
The row vector of coefficients $K_b$ is selected such that the closed-loop matrix $A_b -B_b K_b$ has all negative real eigenvalues.
\end{definition}

\begin{remark}
For a dynamical system with relative degree $r_b = 1$, it becomes a Zeoring Control Barrier Function (ZCBF), with the following constraint 
\begin{align}\label{eq:zcbfConstraint}
\inf \limits_{u\in U} [\pounds_{f} h(x)+&\pounds_{g}h(x)u+\alpha(h(x))] \geq 0, \quad \forall x \in Int(C),\nonumber
\end{align}
where $\alpha$ is a class-K function.
\end{remark}

\subsection{Signal Temporal Logic}
Given a STL formula $\varphi$ with horizon $N$, the syntax of STL is defined as:
\begin{equation}\label{eq:STLsyntax}
    \varphi := \top | \mu | \neg \varphi | \varphi_1 \land \varphi_2 | \varphi_1 \mathcal{U}_{[a,b]} \varphi_2,
\end{equation}
where $\top$ is \textbf{True} symbol in Boolean logic and $\mu$ is the predicate. For each predicate $\mu_i$, we define
\begin{equation} \label{eq:yOut}
    \mu_i := y_i[t] \geq 0,
\end{equation}
with $i=1,...,N_p$, where $N_p$ is the total number of predicates for $\varphi$. The STL semantics is defined as the following:
\begin{equation}
    \begin{aligned}
    &(y,t) \models \mu &\Leftrightarrow& y[t] \geq 0 \\
    &(y,t) \models \mu_1 \land \mu_2 &\Leftrightarrow &(y,t) \models \mu_1 \land (y,t) \models \mu_2 \\
    &(y,t) \models \neg \mu &\Leftrightarrow& \neg ((y,t) \models \mu)\\
    &(y,t) \models \varphi_1 \mathcal{U}_{[a,b]} \varphi_2 &\Leftrightarrow& \exists t' \in [t+a, t+b] \\&s.t. (y,t') \models \varphi_2 \\
    &\land \forall t'' \in [t,t'], (y,t'')\models \varphi_1\\
    \end{aligned}
\end{equation}

We denote the robustness of $\varphi$ as $\rho_{y}^{\varphi}[t]$. The state trajectory satisfies the spec $\varphi$ if and only if $\rho_{y}^{\varphi}[t] \geq 0, \forall t\in[0,t_f]$, where $t_f$ is the end time of the STL horizon. The robustness for each predicate is defined as the following
\begin{equation}
    \begin{aligned}
    &\rho^\mu[y,t] &=& y[t]\\
    &\rho^{\neg \varphi}[y,t] &=& -\rho^{\varphi}[y,t]\\
    &\rho^{\varphi_1 \land \varphi_2} &=& \min(\rho^{\varphi_1} \land \rho^{\varphi_2})\\
    &\rho^{\varphi_1 \mathcal{U}_{[a,b]} \varphi_2}[y,t] &=& \sup_{\tau \in t+[a,b]} (\min(\rho^{\phi_2}[y,\tau]),\inf_{s\in[t,\tau]}\rho^{\varphi_1}[y,s])
    \end{aligned}
\end{equation}

\subsection{Mixed Integer Formulation for STL}
The binary encoding of STL robustness using mixed integer was proposed in \cite{sadraddini2015robust}. The STL formula can be encoded as a set of linear constraints using big-M method, where a sufficiently large number $M$ is introduced to enforce logical constraints. For a STL formula $\varphi$ with horizon $N$, we define a binary variable $z_{\varphi}^t \in \{0,1\}$ with $\varphi \models \top \iff z_{\varphi}^t = 1$ at time step $t$ and $\varphi \models \neg \top$ otherwise. 

For the $i$-th predicate $\mu_i := y_i[t] \geq 0$, we introduce another binary variable $z_{\mu_i}^t \in \{0,1\}$ with $y_i[t] > 0 \iff z_{\mu}(t) = 1$. The big-M constraint is
\begin{align*}
    y[t] \leq M z_{\mu}^t\\
    y[t] \leq M (1-z_{\mu}^t)
\end{align*}

Given a STL formula $\varphi$, we can recursively encode the rest of the logical operators as the following:

\begin{itemize}\label{STL_formulation}
  \item[] \underline{Conjunction}: $z_{\varphi}^t = \wedge_{i=1}^{p} z_{\psi_i}^{t_i}$:
  \begin{align*}
      &z_{\varphi}^t \leq z_{\psi_i}^{t_i},\\
      &z_{\varphi}^t \geq 1-p+\sum \limits_{i=1}^{p} z_{\psi_i}^{t_i},
  \end{align*}
  \item[] \underline{Disjunction}: $z_{\varphi}^t = \vee_{i=1}^{p} z_{\psi_i}^{t_i}$:
  \begin{align*}
      &z_{\varphi}^t \geq z_{\psi_i}^{t_i},\\
      &z_{\varphi}^t \leq \sum \limits_{i=1}^{p} z_{\psi_i}^{t_i},
  \end{align*}
  \item[] \underline{Negation}:$z_{\varphi}^t =  \neg z_{\psi}^{t}$: $z_{\varphi}^t = 1- z_{\psi}^{t}$
  \item[] \underline{Eventually}: $\varphi = \textbf{F}_{[a,b]} \psi$:
  \begin{align*}
    z_{\varphi}^t = \bigvee \limits_{\tau=t+a}^{t+b} z_{\psi_i}^{\tau},
  \end{align*}
  
   \item[] \underline{Always}:$\varphi = \textbf{G}_{[a,b]} \psi$:
   \begin{align*}
    z_{\varphi}^t = \bigwedge \limits_{\tau=t+a}^{t+b} z_{\psi_i}^{\tau},
  \end{align*}
  
   \item[] \underline{Until}:$\varphi = \psi_1 \mathcal{U}_{[a,b]}\psi_2 = \mathbf{G}_{[0,a]} \psi_1 \wedge \mathbf{F}_{[a,b]}\psi_2 \wedge \mathbf{F}_{[a,a]}\psi_1 \mathcal{U}\psi_2$:
   
\end{itemize}
\begin{remark}
   The Big-M encoding for Always ($\textbf{G}$) operator only ensures the signal satisfies a given formula at the sampling time step $[t]$.
   \end{remark}

\section{Problem Statement}
Given the linear system in \eqref{continousSystem}, with initial state $x_0 \in X \subset \mathrm{R^n}$, the goal is to formulate a MIQP problem with the constraints of STL formula and cost function $J(u(t))$. The control sequence is synthesized by solving the corresponding optimization problem offline and the resulting state trajectory $x(t)$ must satisfies a continuous-time STL formula $\varphi$ between $t \in [0,t_f]$ with minimum control effort. 

\section{STL based control with Control Barrier Function}
In section \ref{sec:zoh}, we formally define the ZOH mechanism for implementing the control discretely. In \ref{sec:timeDiscretization}, we show the discretization method and closed-form solution for general linear systems under the ZOH control. Next, given a linear system and linear constraints on system states, we present the CBFs formulation in \ref{sec:CBFconstraintForlinearSys}. And in \ref{sec:CBFlowerbound}, we demonstrate how to obtain the lower bound of a given linear CBF constraint using mixed-integer encoding. In \ref{sec:CBFSTLencoding}, we illustrate certain STL predicates can be encoded as CBF constraints and achieve continuous-time satisfaction. Finally, the MIQP based planner is formally defined in \ref{sec:miqp}.

\subsection{Zeroth-Order Hold Control}
\label{sec:zoh}
The Zeroth-Order Hold (ZOH) control is used in this paper that a generated control signal is held at $t_k$ over a period of time, i.e. $u(s) = u(t_k), \forall s \in [t_k,t_{k+1})$. The sequence of control update time instants $\{t_k\}_{k \in \mathbb{N}}$ is strictly increasing. 

\subsection{Time Discretization}
\label{sec:timeDiscretization}
To discretize the system $\eqref{continousSystem}$, we define a control holding period $\tau$ and $x(t_k)$ as the state at the $k$-th update step, where $k=0,...,N-1$. The closed-form solution for next state $x(t_{k+1})$ is
\begin{equation}\label{eq:discretizeEquation}
\begin{aligned} 
    x(t) =  e^{At}x(t_k) + \int_{t_k}^{t} e^{A(t-s)} ds Bu,& \\
   t_k \leq t \leq t_k+\tau&,
\end{aligned}
\end{equation}
where $u$ is a constant between $[t_0, t_0+\tau]$. 

Under the assumption that $A$ is non-singular, we can utilize the property of matrix exponential and rewrite \eqref{eq:discretizeEquation} to :
\begin{equation}\label{alternativeClosedForm}
x(t) =  e^{At}x(t_0) + A^{-1}(e^{At}-I) Bu.
\end{equation}

Then, we can perform eigen-decomposition on $e^{At}$:
\begin{equation}\label{eq:eigenDecomposition}
    e^{At} = Pe^{Dt}P^{-1},
\end{equation}
where $P \in \mathbb{R}^{n \times n}$ contains the eigenvectors $v_i$ of matrix $A$, and $D\in  \mathbb{R}^{n \times n}$ contains the eigenvalues $\lambda_i$ of $A$ on its diagonal. By combining \eqref{alternativeClosedForm} and \eqref{eq:eigenDecomposition}, we get:
\begin{equation} \label{eq:diagonal}
x(t) = Pe^{Dt}P^{-1}x(0) + A^{-1}(Pe^{Dt}P^{-1}-I)Bu
\end{equation}

If matrix $A$ is singular, we can rewrite \eqref{eq:discretizeEquation} to the following form:
\begin{align} \label{eq:generalForm}
    x(t) = e^{At}x(t_0) + \int_{t_k}^{t} e^{At}V e^{-Js} V^{-1} ds Bu, \\
    t_0 \leq t \leq t_0+\tau, \nonumber
\end{align}
where $e^{-As} = V e^{-Js} V^{-1}$.

\subsection{CBF for Linear Constraint}
\label{sec:CBFconstraintForlinearSys}
Given an initial state $x(t_0) \in \mathbb{R}^n$ and linear system \eqref{continousSystem}, let us consider a linear safety constraint on the $i$-th state variable
\begin{equation} \label{eq:linSafetyConstr}
    h(x_i(t)) = x_i(t) + x_{i,const},
\end{equation}
Based on the closed-form solution of the linear system \eqref{eq:generalForm} and linear safety constraint \eqref{eq:linSafetyConstr}, we can write CBF constraint as
\begin{align} 
\zeta(t) =& \sigma + \sum \limits_{i=1}^n \sum \limits_{j=0}^{s(i)-1} (c_{ij}^{(x)T} x(t_0) e^{\lambda_i t} t^j+{c_{ij}^{(u)T}}u_0 e^{\lambda_i t}t^j) \geq 0,\label{eq:zeta_func}
\end{align}
where $c_{i,j}^{(x)} \in \mathbb{R}^{n}$ is the constant coefficients for $x(t_0)$ and $c_{i,j}^{(u)}$ is the constant coefficients for $u$. The $n$ is the number of Jordan blocks and $s(i)$ is the dimension of the corresponding Jordan block. We define a set of inequalities as the following:
\begin{align} \label{eq:zeta_component}
\zeta^{(x)}(t) &= \beta_{x} + \sigma + \sum \limits_{i=1}^n \sum \limits_{j=0}^{s(i)-1} c_{ij}^{(x)T} x(t_0) e^{\lambda_i t} t^j \geq 0, \\ 
\zeta^{(u)}(t) &= \beta_{u} + \sum \limits_{i=1}^n \sum \limits_{j=0}^{s(i)-1} {c_{ij}^{(u)T}}u_0 e^{\lambda_i t} t^j \geq 0, \nonumber
\end{align}
with $\beta_{x} + \beta_{u} = 0$.

\begin{prop} \label{prop1}
Given the linear system \eqref{continousSystem} and safety constraint \eqref{eq:linSafetyConstr}, if the inequality of \eqref{eq:zeta_component} holds, then the inequality of \eqref{eq:linSafetyConstr} also holds. 
\end{prop}

\begin{proof}
To prove \eqref{eq:zeta_component} $\Rightarrow$ \eqref{eq:zeta_func}, given \eqref{eq:zeta_component} is satisfied, by adding $\zeta^{(x)}(t)$ and  $\zeta^{(u)}(t)$, it is trivial to show that \eqref{eq:zeta_func} is satisfied given the constraint $\beta_{x} + \beta_{u} = 0$.
\end{proof}

\begin{remark}
  If the inequality of \eqref{eq:zeta_func} holds, there must exists a pair of $\beta_x$ and $\beta_u$, such taht $\beta_x+\beta_u = 0$.
\end{remark}
We can write the CBF constraint as the following
\begin{align}
  \zeta(t) =& \sum \limits_{i=1}^n \sum \limits_{j=0}^{s(i)-1} (\zeta_{i,j}^{(x)}(t) + \zeta_{i,j}^{(u)}(t)),
\end{align}
with 
\begin{align*}
\zeta_{i,j}^{(x)}(t) &= \beta_{x} + \sigma + c_{ij}^{(x)T} x(t_0) e^{\lambda_i t} t^j \\ 
\zeta_{i,j}^{(u)}(t) &= \beta_{u} + {c_{ij}^{(u)T}}u_0 e^{\lambda_i t} t^j
\end{align*}

\subsection{CBF Lower Bound through mixed-integer encoding}
\label{sec:CBFlowerbound}
Given a time interval $[t_k,t_k+\tau]$, we have \begin{align}
    \zeta_k (t) &=  \sum \limits_{i=1}^n \sum \limits_{j=0}^{s(i)-1} (\zeta_{k,i,j}^{(x)}(t) + \zeta_{k,i,j}^{(u)}(t)) \\
    &\geq  \min_{\substack{i \in [1,..,n] ;\\ j \in [1,...,s(i)]}} (\zeta_{k,i,j}^{(x)}(t) + \zeta_{k,i,j}^{(u)}(t)).\\ \nonumber
\end{align}

If $\min_{\substack{i \in [1,..,n] ;\\ j \in [1,...,s(i)]}} (\zeta_{k,i,j}^{(x)}(t) + \zeta_{k,i,j}^{(u)}(t)) \geq 0$, then the system trajectory $x(t)$ is safe for $t_k \leq t \leq t_k+\tau$.

Given a time window $[t_k,t_k+\tau]$ and state $x(t_k)$, each term within $\zeta(t)$ can either increase or decrease its overall value for $t>t_k$. We denote the lower bound for the $k$-th time step CBF constraint to be:
\begin{align} \label{eq:zetaLowerBound}
    \zeta_{k,\min}(x(t_k),\tau) = \min_{\substack{i \in [1,..,n] ;\\ j \in [1,...,s(i)]}} \zeta_{k,i,j}^{(x)}(t) +\zeta_{k,i,j}^{(u)}(t),
\end{align}
where $\zeta_{k,\min}(x(t_k),\tau)$ is a lower bound of $\zeta_k (t)$. To find the lower bound, we propose to use Big-M encoding method to remove all the terms that are either monotonically increasing or converging to a positive value. The only terms that left in $\zeta_k(t)$ are strictly decreasing and its minimum value $\zeta_{k,\min}(x(t_k),\tau)$ can be determined given the time bound $[t_k,t_k+\tau]$.

First, We define a set of integer variables $z_{k,i,j}^{(x)} \in \{0,1\}$ and $z_{k,i,j}^{(u)} \in \{0,1\}$ for the following equations at the $k$-th time step
\begin{align}
\zeta_{k,i,j}^{(x)}(t) &= \beta_k^{(x)}+\sigma + \sum \limits_{i=1}^n \sum \limits_{j=0}^{s(i)-1} c_{k,i,j}^{(x)T} x(t_k) e^{\lambda_i t} t^j\\ 
\zeta_{k,i,j}^{(u)}(t) &= \beta_k^{(u)} + \sum \limits_{i=1}^n  \sum \limits_{j=0}^{s(i)-1} c_{k,i,j}^{(u)T}u_k e^{\lambda_i t} t^j
\end{align}

There are a few situations we need to consider. Here are the rules for finding the lower bound of $\zeta_{k,i,j}^{(x,u)}$:
\begin{equation}
    z_{k,i,j}^{(x,u)} =
    \begin{cases}
      0, & c_{k,i,j}^{(x,u)}x(t_0) \geq 0 \wedge \lambda_i \geq 0\\
      0, & c_{k,i,j}^{(x,u)}x(t_0) \geq 0 \wedge \lambda_i \leq 0 \wedge \sigma \geq 0\\
      1, & \text{otherwise}
    \end{cases}
 \end{equation}
For example, if we want to enforce $c_{k,i,j}^{(x,u)}x(t_0) \geq 0 \iff z_{k,i,j}^{(x,u)} = 0$, the following mixed integer encoding is used:
 \begin{align*}
     c_{k,i,j}^{(x,u)}x(0) \leq M (1-z_{k,i,j}^{(x,u)}), \\ 
     -c_{k,i,j}^{(x,u)}x(0) \leq Mz_{k,i,j}^{(x,u)},
 \end{align*}
 For $\lambda_i \geq 0 \iff z_{k,i,j}^{(x,u)} = 0$, we have the following
 \begin{align*}
   \lambda_i \leq M (1-z_{k,i,j}^{(x,u)}), \\ 
   -\lambda_i \leq M z_{k,i,j}^{(x,u)},
 \end{align*}
where $M$ is a sufficiently large number. 

For all terms such that $z_{k,i,j}^{(x,u)} = 1$, we need to ensure $\zeta_{k,\min}(x(t_k),\tau)$ is positive given some time window.  Consider the CBF lower bound \eqref{eq:zetaLowerBound}, for time $\forall t \in [t_k, t_k+\tau]$. We have the following cases: 
\begin{case}
For $c_{k,i,j}^{(x,u)T} x(t_k) < 0, \lambda_i \geq 0, j > 0$, $\min \zeta_{k,i,j}^{(x,u)}(t) = \zeta_{k,i,j}^{(x,u)}(\tau)$. 
\end{case}

\begin{case}
For $c_{k,i,j}^{(x,u)T} x(t_k) < 0, \lambda_i > 0, j = 0$, $\min \zeta_{k,i,j}^{(x)}(t) = c_{k,i,j}^{(x)T} x(t_k) + \sigma + \beta_{k}^{(x)}$ and $\min \zeta_{k,i,j}^{(u)}(t) = c_{k,i,j}^{(u)T} x(t_k) + \beta_{k}^{(u)}$
\end{case}

\begin{case}
For $c_{k,i,j}^{(x,u)T} x(t_k) < 0, \lambda_i < 0, j > 0$, $\min \zeta_{k,i,j}^{(x,u)}(t) = \zeta_{k,i,j}^{(x,u)}(-\frac{j}{\lambda_i})$. 
\end{case}

By following the rules above, we can add the lower bound of the CBF constraint $\min \zeta_{k,i,j}^{(x,u)}(t)$ as linear constraints to our optimization problem. 

\subsection{CBF constraints with STL predicates}
\label{sec:CBFSTLencoding}
The idea of using CBF for safety can be carried over to ensure STL satisfactory in continuous time. For example, we want to ensure $\phi = \textbf{G}_[t_3,t_4] \mu_1$ where $\mu_1 := x_2[t] - 3$. Based on the integer encoding method above, assuming the system starts at $t=0$ and $z_{\phi}^t = 1$, we have
\begin{align*}
    &z_{\phi}^t \leq z_{\mu_1}^{t_3},\\
    &z_{\phi}^t \leq z_{\mu_1}^{t_4},\\
    &z_{\phi}^t \geq -1+z_{\mu_1}^{t_3}+z_{\mu_1}^{t_4}
\end{align*}

The formulation above ensures $z_{\mu_1}^{t_3} = z_{\mu_1}^{t_4} =1$, which implies $x_2[t_3] \geq 3 $ and $x_2[t_4] \geq 3 $. But we cannot draw a conclusion in between $[t_3,t_4]$. To ensure $x_2(t) \models \mu_1, \forall t \in [t_3,t_4]$, we need to define the following proposition. 

\begin{prop}
Given a linear system in \eqref{continousSystem} and always operator $\textbf{G}$.  For some linear STL predicate $\mu$, the trajectory $x(t)$ satisfies formula $\varphi = G_{[a,b]} (\mu:=h_{\mu}(x))$ if and only if

\begin{align} 
h_{\mu}(x(a)) \geq 0, \label{eq:h_mu} \\
\zeta_{k,\min}(x(t_k),\tau) \geq 0, \label{eq:zeta_min}
\end{align}
where $\tau := b - a$ and $t_k = a$.
\end{prop}

\begin{proof}
Assume \eqref{eq:h_mu} is satisfied, that implies the system trajectory $x(t)$ is within the safety set defined by $h_{\mu}(x(t))$ at $t=a$. Based on the definition of \eqref{eq:zetaLowerBound}, the inequality \eqref{eq:zeta_min} implies
\begin{align*}
    \min_{\substack{i \in [1,..,n] ;\\ j \in [1,...,s(i)]}} \zeta_{k,i,j}^{(x)}(t) +\zeta_{k,i,j}^{(u)}(t) \geq 0,  a \leq t \leq b.
\end{align*}
With Proposition \ref{prop1}, the inequality above further implies $\zeta(t) \geq 0, \forall t \in [a,b]$.

Based on the definition of the ECBF (ZCBF) \eqref{eq:ecbfConstraint}, the system is \textit{forward invariant}, i.e., $x(t)$ stays within the safety set defined by $h_{\mu}$ for $\tau$ period and $x(t) \models \mu, a \leq t \leq b$.
\end{proof}

\subsection{Optimization Problem}
\label{sec:miqp}
The trajectory planning problem is formulated as the following MIQP: 
\begin{equation} 
\begin{aligned}
& \underset{\textbf{u},\textbf{x}}{\text{min}}
& & \int_{0}^{t_f} u^Tu \\
& \text{s.t.}
& & x[k+1] = A_k x[k] + B_k u_k\\
& & & x(t) 	\models \varphi \\
& & & u_{l} \leq u_k \leq u_{u}.\\
& & & z_{\varphi}, z_{k,i,j}^{(x)}, z_{k,i,j}^{(u)}\in \{0,1\} \\
& & & k = 0,...,N-1. \\
& & & t \in [0,t_f], i=[1,...,n],j=[1,...,s(i)]
\end{aligned}
\label{eq:optimizationProblem}
\end{equation}

Note the dynamics constraints are encoded using direct multiple-shooting method based on the closed-form solution of linear dynamics. The goal is to solve the offline optimal control problem \eqref{eq:optimizationProblem} above and obtain an optimal trajectory ($\textbf{x}^*$,$\textbf{u}^*$), with $\textbf{x}^*=x_1^*,...,x_N^*$ and $\textbf{u}^*=u_0^*,...,u_{N-1}^*$, under the mixed-integer constraints from the STL specifications and the CBF constraints. We assume a total number of update steps $N$ and individual update period $\tau_k := \frac{t_f}{N}$. For each time step, given a fixed time interval $[t_k,t_k+\tau_k]$, the solver will attempt to find a feasible pair of $(u_k, x(t)), t_k \leq t \leq t_k + \tau_k$ that satisfies all constraints for $k=1,...,N$. 

\section{Numerical Examples}
Let us consider a double integrator system that has the following form:
\begin{align} \label{eq:doubleIntSys}
\left[\begin{matrix}\dot{x_1}(t) \\ \dot{x_2}(t) \end{matrix} \right]= \left[ \begin{matrix} 0&1\\ 0&0 \end{matrix} \right] \left[\begin{matrix} x_1(t) \\ x_2(t) \end{matrix} \right]+ \left[ \begin{matrix} 0\\ 1 \end{matrix} \right] u.
\end{align}
Based on \eqref{eq:generalForm}, the closed-form solution for this system is
\begin{equation}
    \left[\begin{matrix}x_1(t) \\ x_2(t) \end{matrix} \right]= \left[ \begin{matrix} 1&t\\ 0&1 \end{matrix} \right] \left[\begin{matrix} x_1(t_k) \\ x_2(t_k) \end{matrix} \right]+ \left[ \begin{matrix} t& \frac{t^2}{2}\\ 0&t \end{matrix} \right] \left[ \begin{matrix} 0\\ 1 \end{matrix} \right] u.
\end{equation}

Given horizon $N$ and final time $t_f$ in seconds. We define a sequence of discrete time steps $\{t_k|k = 0,..., N, t_0 = 0, t_N = t_f\}$. Starting at initial state $x(t_0) = [1,-1]^T$, the goal is to generate a sequence of control $\textbf{u} = u_0,...,u_{9}$ that ensure the system satisfy a continuous-time STL specification $\varphi$.

In example 1, we enforce a velocity constraint using CBF on the system while its position is required to oscillate between a specific interval. In example 2, we combine safety requirements and logic operator to achieve disjunctive ($\vee$) safety specification. A comparison is made between our method and the method used in \cite{sadraddini2015robust}. In example 3, we further demonstrate how our method can be used for asynchronous update given the time constraints from STL specification and sampling rate from the system are not the same. All examples were formulated as MIQPs using Gurobi \cite{gurobi} and solved on a i5-8259U CPU. 

\subsection{Example 1: STL based planning with Safety Constraints}
Consider the following STL formula with $t_f = 2s$:
\begin{align}
\varphi_1 = &F_{[0.2s,0.8s]} (x_1(t) <= -2) \wedge\\ 
&F_{[1s,1.4s]} (x_1(t) >= 2) \wedge \nonumber\\ 
&F_{[1.6s,2s]} (x_1(t) <= -2), \nonumber
\end{align}
with the safety requirement of $ -10 < x_2(t) < 10, \forall t$. The goal is to steer the state $x_1(t)$ to oscillate between $-2$ and $2$ within some real time interval. 
\begin{remark}
  The safety requirement can be viewed as STL predicates in continuous time with the following form: $\textbf{G}_{[0,t_f]} (h_{\varphi_1,1}(x))\wedge(h_{\varphi_1,2}(x))$, with $h_{\varphi_1,1}(x) = x_2 -10$ and $h_{\varphi_1,2}(x) = -x_2 + 10$.
\end{remark}
The elapsed time for solving the MIQP is 0.02 seconds.

\begin{figure}[H]
 \includegraphics[width=0.52\textwidth]{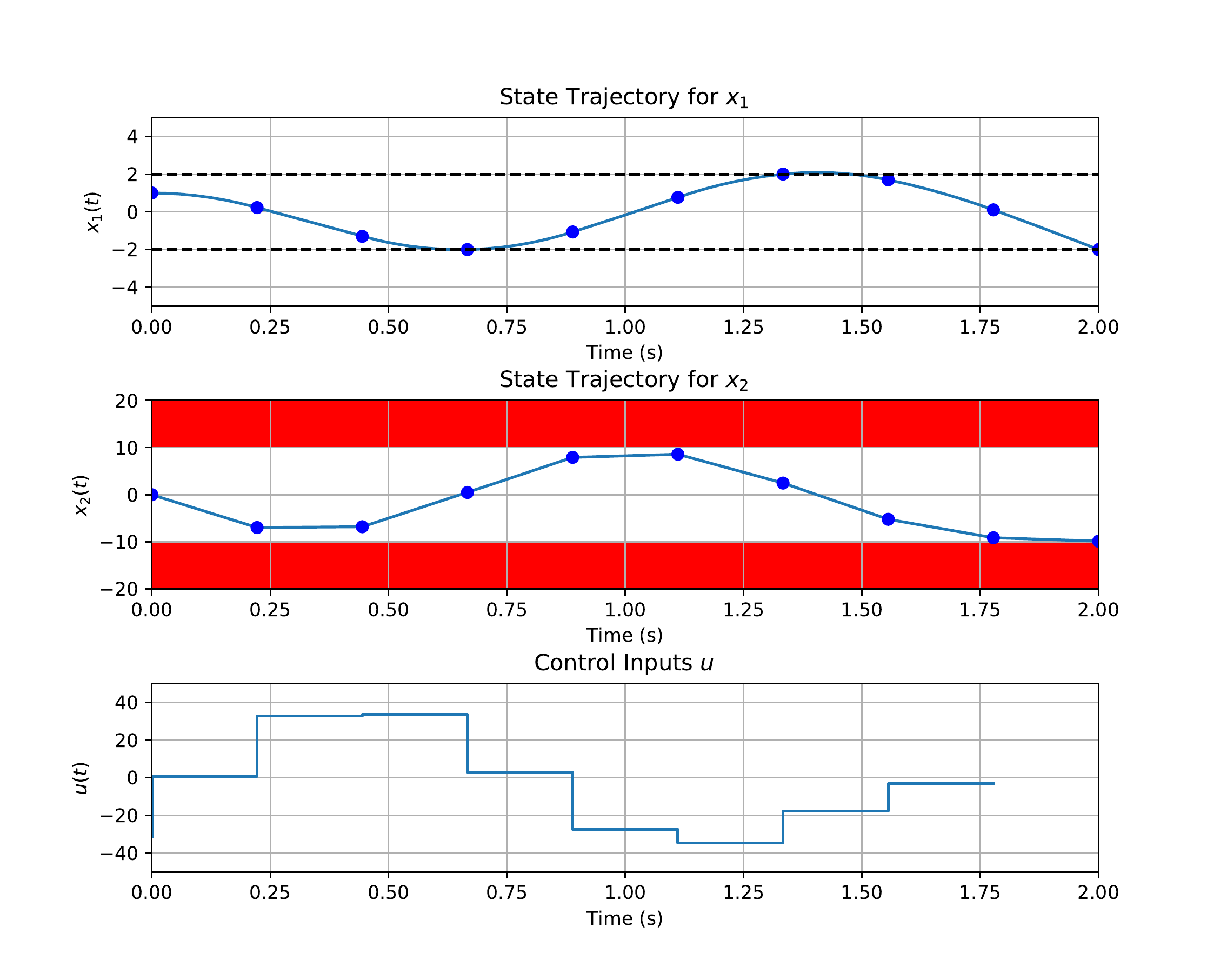}
  \caption{Velocity constraints using CBFs}
 \end{figure}

\subsection{Example 2: Safety Sets as the STL Predicates with Disjunction}
\label{sec:example2}
Consider the following two-dimensional double-integrator system with $t_f = 1s$:
\begin{align} \label{eq:2ddoubleIntSys}
\left[\begin{matrix}\dot{x_1}(t) \\ \dot{x_2}(t) \\ \dot{x_3}(t) \\ \dot{x_4}(t)  \end{matrix} \right]= \left[ \begin{matrix} 0&1&0&0\\ 0&0&0&0\\0&0&0&1\\0&0&0&0 \end{matrix} \right] \left[\begin{matrix} x_1(t) \\ x_2(t)\\ x_3(t)\\ x_4(t) \end{matrix} \right]+ \left[ \begin{matrix} 0&0\\ 1&0 \\0&0 \\0&1 \end{matrix} \right] \left[ \begin{matrix} u_1 \\ u_2\end{matrix} \right],
\end{align}
with corresponding closed-form solution:
\begin{align} \label{eq:2ddoubleIntSys}
\left[\begin{matrix}x_1(t) \\ x_2(t) \\ x_3(t) \\ x_4(t)  \end{matrix} \right]= \left[ \begin{matrix} 1&t&0&0\\ 0&1&0&0\\0&0&1&t\\0&0&0&1 \end{matrix} \right] \left[\begin{matrix} x_1(t) \\ x_2(t)\\ x_3(t)\\ x_4(t) \end{matrix} \right]+ \left[ \begin{matrix} \frac{t^2}{2}&0\\ t&0 \\0&\frac{t^2}{2} \\0&t \end{matrix} \right] \left[ \begin{matrix} u_1 \\ u_2\end{matrix} \right],
\end{align}

The extension of our proposed approach is to directly encode safety sets as part of the STL specification. In other words, the safety sets follow the STL syntax \eqref{eq:STLsyntax}. In this example, we demonstrate the disjunctive logic operator on safety sets (e.g. $h_1(x) \geq 0 \vee h_2(x) \geq 0$). 

We would like the system trajectory to satisfy the following STL formula: 
\begin{align} \label{eq:STL_example_spec}
\varphi_2 := &\textbf{F}_{[0.1s,0.6s]} (x_1(t) \leq -0.5 \wedge x_3(t) \geq 0.5) \wedge\nonumber\\ &\textbf{F}_{[0.7s,1.0s]} (x_1(t) \geq 1 \wedge x_3(t) \geq 1) \wedge\\ \nonumber &\textbf{G}_{[0s,1.0s]} (x_1(t)\geq 0 \vee x_3(t) \geq 0),\\ \nonumber &t\in {[0,t_f]}. \end{align}

The predicate with always operator can be viewed as a disjuction of two safety sets by defining $h_{\varphi_2,1}(x) = x_1(t)$ and $h_{\varphi_2,2}(x) = x_3(t)$ with the following equivalence: 
\begin{align*}
\textbf{G}_{[0,t_f]} (x_1(t)\geq 0 \vee x_3(t) \geq 0) \models \top \iff \\ (h_{\varphi_2,1}(x(t)) \geq 0) \vee (h_{\varphi_2,2}(x(t)) \geq 0), t\in[0,t_f]
\end{align*}

Note that $h_{\varphi_2,1}$ and $h_{\varphi_2,2}$ have relative degrees of $r_b = 2$, we need to use ECBF \eqref{eq:ecbfConstraint} constraints:
\begin{align}\label{eq:example2_cbf_constraint}
    &\zeta_{\varphi_2,1} = \pounds_{f}^{2} h_{\varphi_2,1}(x)+\pounds_{g}\pounds_{f} h_{\varphi_2,1}(x)u + k_1 h_{\varphi_2,1} + k_2 \dot{h}_{\varphi_2,1}(x)\\ \nonumber
    &\zeta_{\varphi_2,2} =  \pounds_{f}^{2} h_{\varphi_2,2}(x)+\pounds_{g}\pounds_{f} h_{\varphi_2,1}(x)u + k_1 h_{\varphi_2,2} + k_2 \dot{h}_{\varphi_2,1}(x), \nonumber
\end{align}
where the system is within safety set $h_{\varphi_2,1}$ or $h_{\varphi_2,2}$, if $\zeta_{\varphi_2,1} \geq 0$ or $\zeta_{\varphi_2,1} \geq 0$ respectively. We can obtain the lower bounds $\min \zeta_{\varphi_2,1}(x,t)$,$\min \zeta_{\varphi_2,2}(x,t)$ using our proposed mixed-integer encoding method in \ref{sec:CBFlowerbound} and apply the following constraints:
\begin{align*}
    &\min \zeta_{k,\varphi_2,1}(x,t) \geq 0\\
    &\min \zeta_{k,\varphi_2,2}(x,t) \geq 0\\
    &t_k \leq t \leq t_k+\tau, k=1,...,N
\end{align*}

\begin{figure}[htbp]
\begin{minipage}[t]{0.50\linewidth}
    \includegraphics[width=\linewidth]{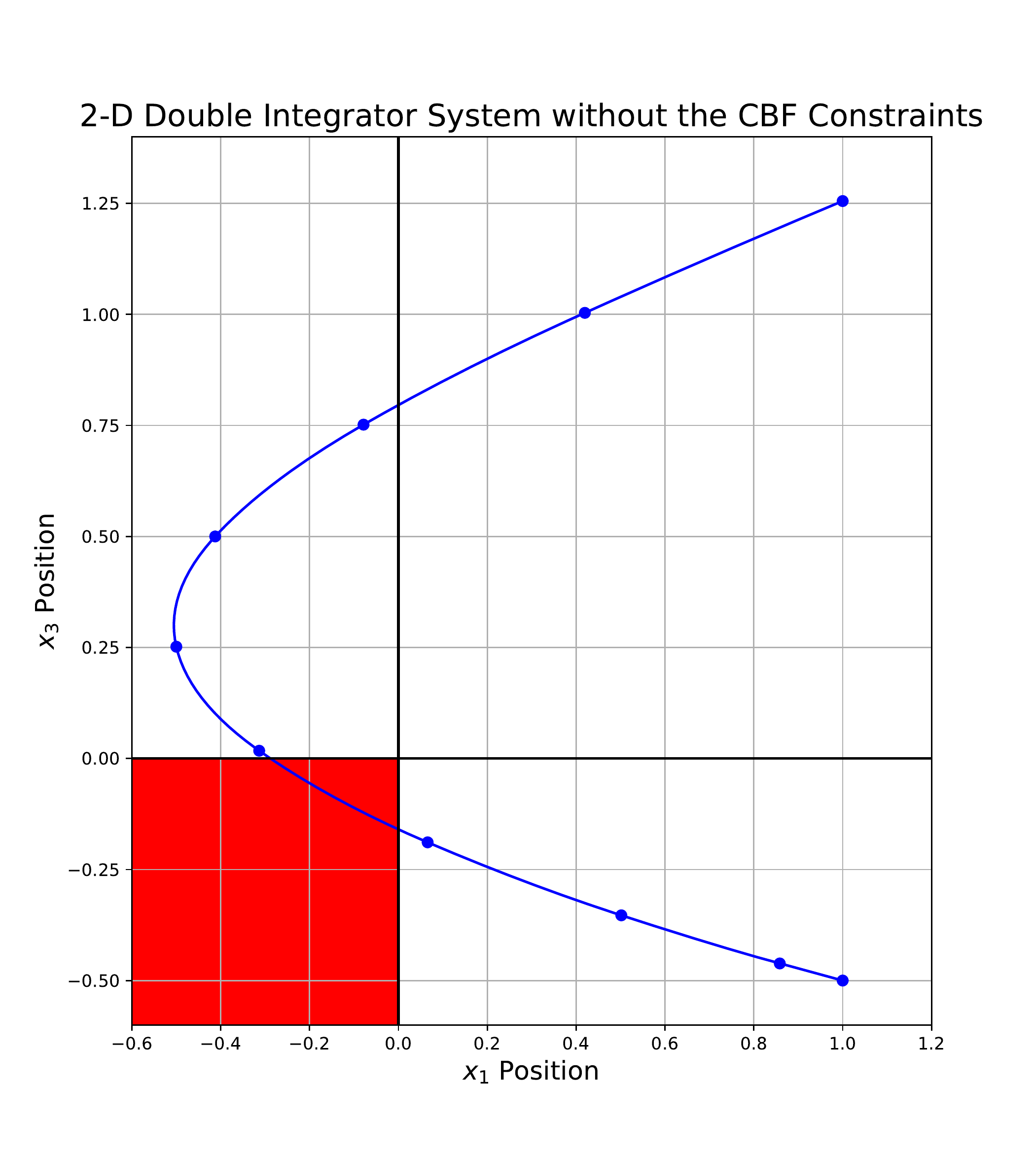}
    \caption{Case 1}
    \label{fig:f2}
\end{minipage}%
    \hfill%
\begin{minipage}[t]{0.50\linewidth}
    \includegraphics[width=\linewidth]{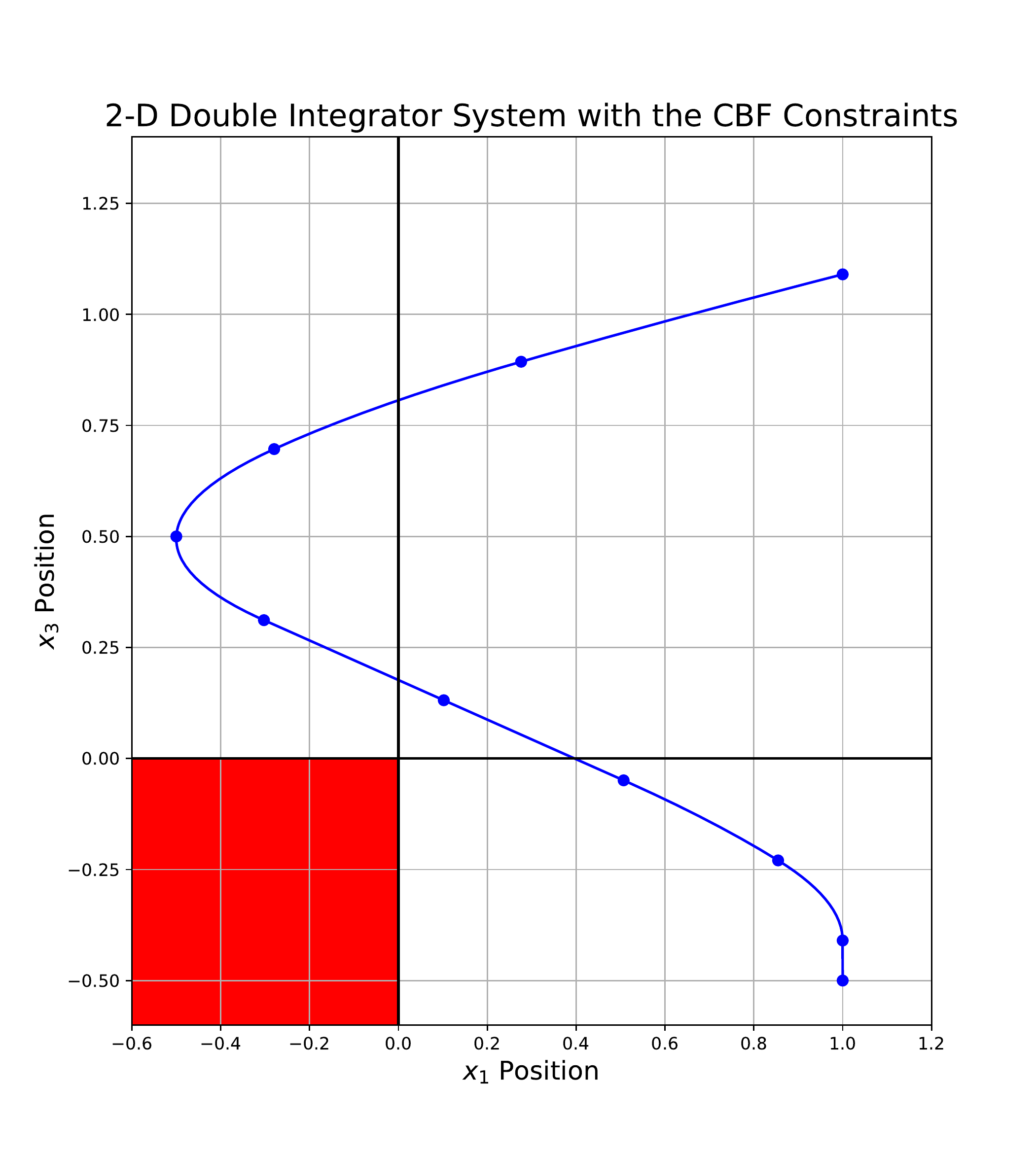}
    \caption{Case 2}
    \label{fig:f3}
\end{minipage} 
\end{figure}

In Case 1, the STL specification $\varphi_2$ is encoded with integer encoding method. Notice how the discrete states trajectory satisfies $\varphi_2$ in Figure $\ref{fig:f2}$, but the continuous-time state trajectory clearly violates the specification. In Case 2, the CBFs constraints \eqref{eq:example2_cbf_constraint} are used in the MIQP and resulting trajectory (Shown in Figure \ref{fig:f3}) satisfies $\varphi_2$ in continuous time. The optimization problems in Case 1 and Case 2 are solved in 0.063s and 0.12s respectively.
\vspace{-0.5pt}
\subsection{Example 3: continuous-time STL under asynchronous time lines}
\label{sec:example3}
Finally, we illustrate the continuous-time STL formula can be satisfied under the limitation of fixed system sampling rate using our proposed method. Let us consider the case where the system has a fixed sampling rate but some STL predicates defined an time interval in-between two sampling instances. Such formulation might leads to the issue of violation on the specification by using the integer encoding method without CBF constraints. 

Consider the following example with the same double integrator system \eqref{eq:doubleIntSys}, with STL formula $\varphi_3 = \textbf{G}_{[0.63s, 0.80s]}  (x_2(t) >= 3) \wedge \textbf{F}_{[1.4s, 2.0s]} (x_2(t) <= -4)$.

We approach the problem by considering two systems under different time sequences, namely \textit{Simulated System} and \textit{Real System}, with notation $sim$ and $real$ respectively. Next, we define time sequences $\{t_{sim}^{k_{sim}}\}$ and $\{t_{real}^{k_{real}}\}$, with $k_{sim}=1,...,N_{sim}$ and $k_{real}=1,...,N_{real}$ (Fig \ref{fig:example3timeline}). In addition, we define the control sequences as $u_{sim}$ and $u_{real}$ accordingly.

\begin{figure}[H]
\centering
\includegraphics[width=0.9\linewidth]{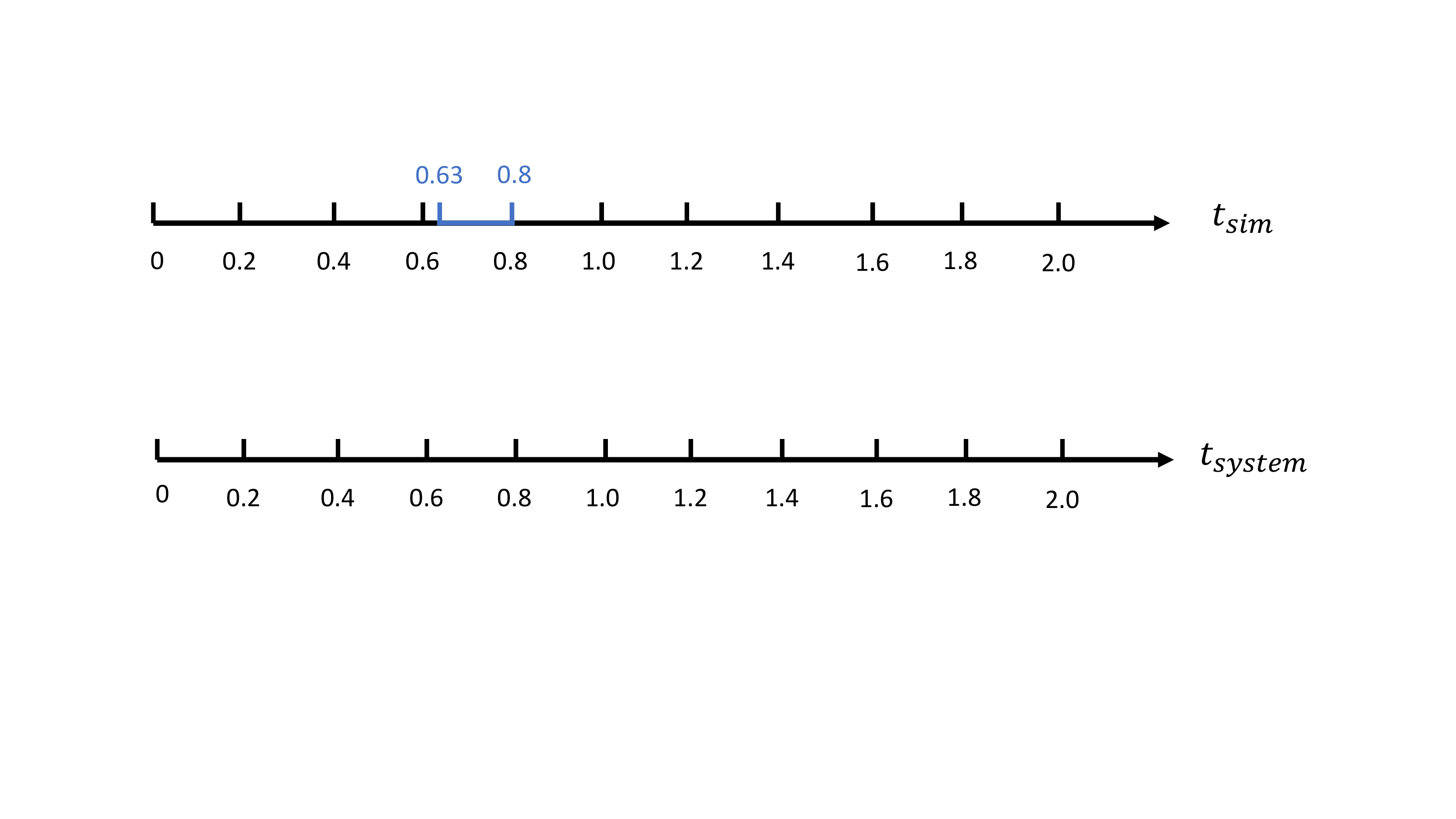}
\caption{Time lines for $t_{sim}$ and $t_{real}$}
\label{fig:example3timeline}
\end{figure}

For this example, we have $N_{sim} = 11$ and $N_{real} = 12$ because the Always $\textbf{G}$ predicate is defined in-between two sampling instances, i.e., $[0.63s,0.8s]$. Notice under the system update constraint, it only allows to apply control on the sampling time instants (i.e., $t=0.2s,0.4s,...,2.0s$). To resolve this asynchronous issue, we solve the corresponding MIQP under $t_{sim}$ with additional constraint $u_{sim}[t=0.6] = u_{sim}[t=0.63]$.

For $\textbf{G}_{[0.63s, 0.80s]}  (x_2(t) >= 3)$, the following CBF constraint is added:
\begin{equation}\label{eq:example3CBF}
 \zeta_{\min,\varphi_3}(x_2(t=0.63),\tau) \geq 0,
\end{equation}
where $\tau = 0.8s-0.63s = 0.17s$.

Finally, we can apply the synthesized control $u_{sim}$ under $t_{real}$ and still satisfies $\varphi_3$ in continuous time. 

\begin{figure}[H]
\centering
\includegraphics[width=1.0\linewidth]{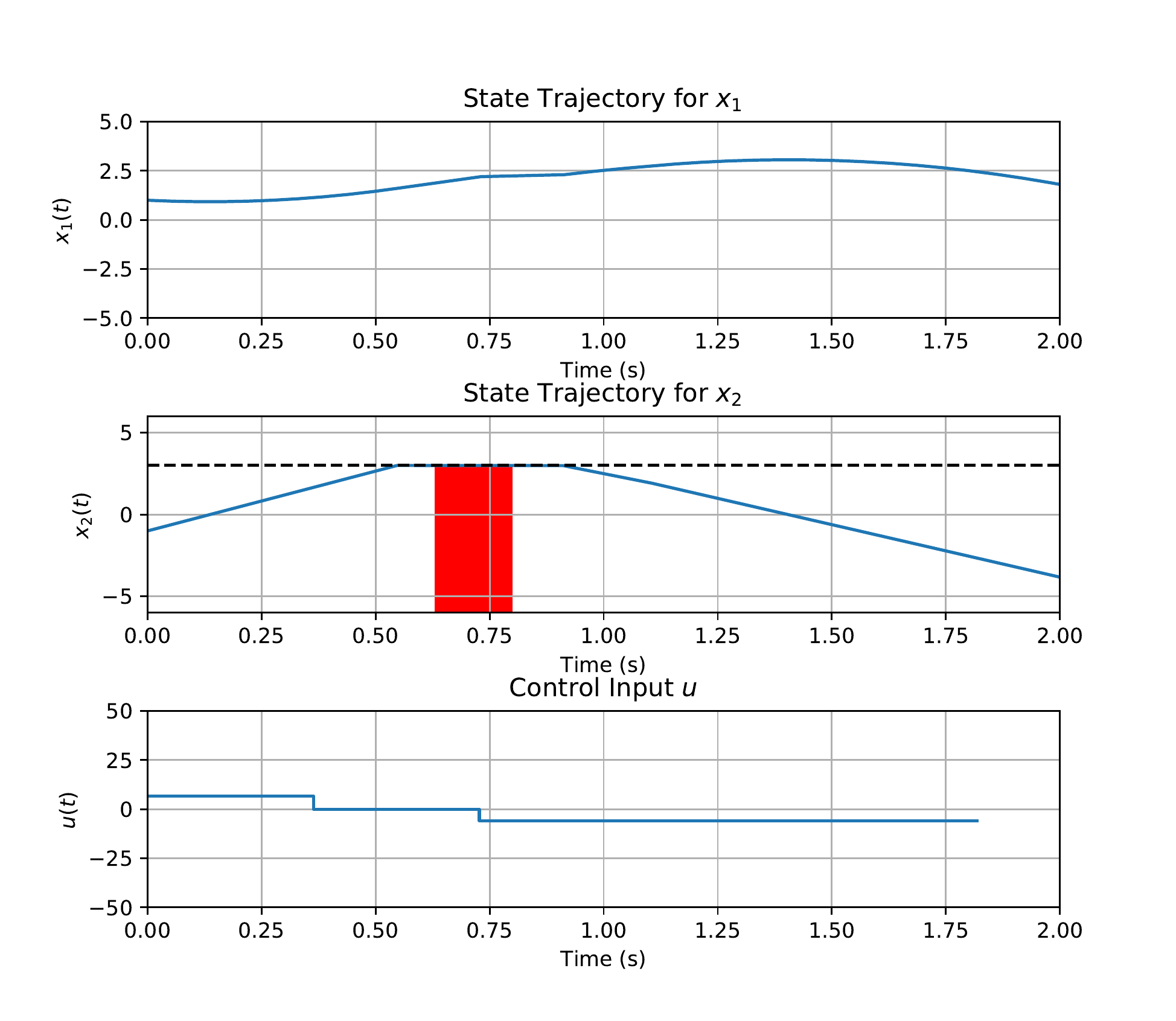}
\caption{Time lines for $t_{sim}$ and $t_{real}$}
\label{fig:example3}
\end{figure}

In Fig \ref{fig:example3}, we demonstrate $\textbf{G}_{[0.63s, 0.80s]}$ is satisfied by define an unsafe region (Red) using \eqref{eq:example3CBF}. Although the CBF constraint only active at $t=0.63$, the additional constraint, i.e.,$u_{sim}[t=0.63] = u_{real}[t=0.6]$, ensures we can directly assign $u_{real}[t=0.6] := u_{sim}[t=0.63]$. The resulting trajectory from the real system still satisfy the original specification. The MIQP is solved in 0.042s.

\section{Conclusion}
We have developed a complete framework of trajectory planning under the real-valued space and time constraints from the STL specification. In addition, the notion of encoding Always \textbf{G} opeartor using the lower bound of CBF is introduced to overcome the drawbacks from discrete-time STL robustness based control. The future work is to extend our proposed method to a Model Predictive Control (MPC) so that we can obtain an feedback controller for real-world applications. In addition, we are also interested in having an adaptive control where a STL is satisfied in continuous-time but with minimum number of controller updates. 

\bibliographystyle{IEEEtran}
\bibliography{IEEEabrv,mybib}

\end{document}